\documentclass{llncs}

\usepackage{amsmath,amssymb}
\usepackage{graphicx}
\usepackage{fullpage}
\usepackage[labelsep=period,labelfont=bf,size=small,compatibility=false]{caption}
\usepackage[font={small},labelfont=default]{subfig}
\usepackage[plainpages=false]{hyperref}
\usepackage{paralist}
\usepackage[ruled,vlined]{algorithm2e}
\SetKwInOut{Input}{Input} \SetKwInOut{Output}{Output}

\newcommand{\fpfrs}{\textsc{Fan-Planarity with Fixed Rotation System}}
\newcommand{\threepart}{\textsc{3-Partition}}

\title{On the Recognition of Fan-Planar and\\ Maximal Outer-Fan-Planar Graphs
\thanks{This work started at the \emph{Bertinoro Workshop on Graph Drawing 2014}. We thank the organizers and the participants of the workshop for the useful discussions on this topic. The work of M.A.~Bekos is implemented within the framework of the Action ``Supporting Postdoctoral Researchers'' of the Operational Program ``Education and Lifelong Learning'' (Action's Beneficiary: General Secretariat for Research and Technology), and is co-financed by the European Social Fund (ESF) and the Greek State. L.~Grilli was partly supported by the MIUR project AMANDA ``Algorithmics for MAssive and Networked DAta'', prot. 2012C4E3KT\_001. S. Hong was partly supported by her ARC Future Fellowship and Humboldt Fellowship.}
}

\author{
M.~A.~Bekos\inst{1},
S.~Cornelsen\inst{2},
L.~Grilli\inst{3},
S.-H.~Hong\inst{4},
M.~Kaufmann\inst{1}
}

\date{}

\institute{%
    Wilhelm-Schickard-Institut f\"ur Informatik, Universit\"at T\"ubingen, Germany\\
    \email{$\{$bekos,mk$\}$@informatik.uni-tuebingen.de} \and
    Dept. of Computer and Information Science, University of Konstanz, Germany\\
    \email{sabine.cornelsen@uni-konstanz.de} \and
    Dipartimento di Ingegneria, Universit\`a degli Studi di Perugia, Italy\\
    \email{luca.grilli@unipg.it} \and
    School of Information Technologies, University of Sydney, Australia\\
    \email{shhong@it.usyd.edu.au}
}

\begin{document}
\maketitle

\begin{abstract}
\emph{Fan-planar} graphs were recently introduced as a
generalization of \mbox{$1$-\emph{planar}} graphs. A graph is
\emph{fan-planar} if it can be embedded in the plane, such that each
edge that is crossed more than once, is crossed by a bundle of two
or more edges incident to a common vertex. A graph is
\emph{outer-fan-planar} if it has a fan-planar embedding in which
every vertex is on the outer face. If, in addition, the insertion of
an edge destroys its outer-fan-planarity, then it is \emph{maximal
outer-fan-planar}.

In this paper, we present a polynomial-time algorithm to test
whether a given graph is \emph{maximal outer-fan-planar}. The
algorithm can also be employed to produce an outer-fan-planar
embedding, if one exists. On the negative side, we show that testing
fan-planarity of a graph is NP-hard, for the case where the
\emph{rotation system} (i.e., the cyclic order of the edges around
each vertex) is given.
\end{abstract}

\section{Introduction}
\label{sec:introduction}

A \emph{simple drawing} of a graph is a representation of a graph in
the plane, where each vertex is represented by a point and each edge
is a Jordan curve connecting its end-points such that no edge
contains a vertex in its interior, no two edges incident to a common
end-vertex cross, no edge crosses itself, no two edges meet
tangentially, and no two edges cross more than once.

An important subclass of drawn graphs is the class of planar
graphs, in which there exist no crossings between edges. Although
planarity is one of the most desirable properties when drawing a
graph, many real-world graphs are in fact non-planar.

On the other hand, it is widely accepted that edge crossings
have negative impact on the human understanding of a graph
drawing~\cite{DBLP:journals/iwc/Purchase00} and simultaneously it is NP-complete in
general to find drawings with minimum number of edge
crossings~\cite{gj-cigtnpc-79}. This motivated the
study of ``almost planar'' graphs which may contain crossings as long as they do not violate
some prescribed forbidden crossing patterns. Typical examples of
such graphs include $k$-planar graphs~\cite{MR0187232}, $k$-quasi
planar graphs~\cite{DBLP:journals/combinatorica/AgarwalAPPS97}, RAC
graphs~\cite{DBLP:journals/tcs/DidimoEL11} and fan-crossing free
graphs~\cite{DBLP:conf/isaac/CheongHKK13}.

Fan-planar graphs were recently introduced in the same
context~\cite{DBLP:journals/corr/KaufmannU14}. Typically, a
\emph{fan-planar drawing} of graph $G=(V,E)$ is a simple drawing which
allows for more than one crossing on an edge $e \in E$ if and only if
the edges that cross $e$ are incident to a common vertex on the same
side of $e$. Such a crossing is called \emph{fan-crossing} (see
Fig.\ref{fig:fan-crossing}). An equivalent definition can be stated by
means of forbidden crossing patterns; see Fig.~\ref{fig:fanplanar1},
\ref{fig:fanplanar2} and \ref{fig:trianglecrossing}. A graph is
\emph{fan-planar} if it admits a fan-planar drawing. Note that the
class of fan-planar graphs is in a sense the complement of the class
of fan-crossing free graphs~\cite{DBLP:conf/isaac/CheongHKK13}, which
simply forbid fan-crossings.

\begin{figure}[t]
  \centering
  \begin{minipage}[b]{.21\textwidth}
    \centering
    \subfloat[\label{fig:fan-crossing}{Fan-crossing}]
    {\includegraphics[width=\textwidth,page=1]{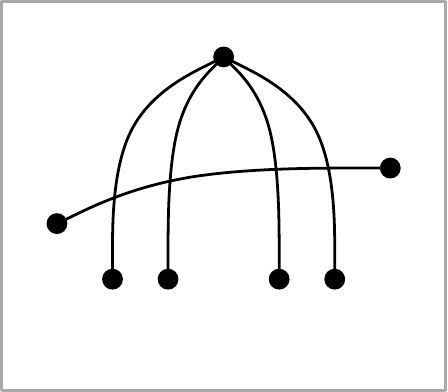}}
  \end{minipage}
  \begin{minipage}[b]{.21\textwidth}
    \centering
    \subfloat[\label{fig:fanplanar1}{Forbidden pattern I}]
    {\includegraphics[width=\textwidth,page=2]{fan-planar}}
  \end{minipage}
  \begin{minipage}[b]{.21\textwidth}
    \centering
    \subfloat[\label{fig:fanplanar2}{Forbidden pattern II}]
    {\includegraphics[width=\textwidth,page=3]{fan-planar}}
  \end{minipage}
  \begin{minipage}[b]{.21\textwidth}
    \centering
    \subfloat[\label{fig:trianglecrossing}{Triangle crossing}]
    {\includegraphics[width=\textwidth,page=4]{fan-planar}}
  \end{minipage}
  \caption{ (taken from \cite{DBLP:journals/corr/KaufmannU14})(a)~Illustration of a fan-crossing.
  (b)~Forbidden crossing pattern I: An edge cannot be crossed by two independent edges.
  (c)~Forbidden crossing pattern II: An edge cannot be crossed by two edges having their common end-point on different sides of it.
  (d)~Forbidden crossing pattern II implies that an edge cannot be crossed by three edges forming a triangle.}
  \label{fig:fanplanar}
\end{figure}

To the best of our knowledge the only known result for this particular
class of graphs is due to Kaufmann and
Ueckerdt~\cite{DBLP:journals/corr/KaufmannU14}, who showed that a
fan-planar graph on $n$ vertices cannot have more than $5n-10$ edges
and this bound is tight. An \emph{outer-fan-planar drawing} is a
fan-planar drawing in which all vertices are on the outer face. A
graph is \emph{outer-fan-planar} if it admits an outer-fan-planar
drawing. An outer-fan-planar graph is maximal outer-fan-planar if
adding any edge to it yields a graph that is not
outer-fan-planar. Note that the forbidden pattern II is irrelevant for
outer-fan-planarity. Our main contribution is a polynomial time
algorithm for the recognition of maximal outer-fan-planar graphs
and significant insights in their structural properties (see
Section~\ref{sec:outerfanplanar}).
We also prove that the general
fan-planar problem is NP-hard, for the case where the \emph{rotation
system} (i.e., the circular order of the edges around each vertex)
is given (see Section~\ref{sec:NPhard}). The question how to test
(non-maximal) outer-fan-planarity efficiently is left open.

\subsection{Related Work}
\label{sec:related_work}

As already stated, $k$-planar graphs~\cite{MR0187232}, $k$-quasi
planar graphs~\cite{DBLP:journals/combinatorica/AgarwalAPPS97}, RAC
graphs~\cite{DBLP:journals/tcs/DidimoEL11} and fan-crossing free
graphs~\cite{DBLP:conf/isaac/CheongHKK13} are closely related to the
class of graphs we study. A graph is \emph{$k$-planar}, if it can be
embedded in the plane with at most $k$ crossings per edge.
Obviously, $1$-planar graphs are also fan-planar. A $1$-planar graph
with $n$ vertices has at most $4n-8$ edges and this bound is
tight~\cite{MANA:MANA3211170125,DBLP:journals/dm/FabriciM07,DBLP:journals/combinatorica/PachT97}.
Grigoriev and
Bodlaender~\cite{DBLP:journals/algorithmica/GrigorievB07}, and,
independently Kohrzik and Mohar~\cite{DBLP:journals/jgt/KorzhikM13}
proved that the problem of determining whether a graph is 1-planar
is NP-hard and remains NP-hard, even if the deletion of an edge
makes the input graph
planar~\cite{DBLP:journals/corr/abs-1203-5944}.

On the positive side, Eades et
al.~\cite{DBLP:journals/tcs/EadesHKLSS13} presented a linear time
algorithm for testing {\em maximal} 1-planarity of graphs with a
given rotation system. Testing outer-$1$-planarity of a graph can be
solved in linear time, as shown independently by  Auer et
al.~\cite{DBLP:conf/gd/AuerBBGHNR13} and Hong et
al.~\cite{DBLP:conf/gd/HongEKLSS13}.  It is worth to note that an
outer-$1$-planar graph is always
planar~\cite{DBLP:conf/gd/AuerBBGHNR13}, while this is not true in
general for outer-fan-planar graphs. Indeed, the complete graph
$K_5$ is outer-fan-planar, but not planar.

The well-known Fary's theorem~\cite{Fary} proved that every plane
graph admits a straight-line drawing. However,
Thomassen~\cite{DBLP:journals/jgt/Thomassen88a} presented two
forbidden subgraphs for straight-line drawings of 1-plane graphs.
Hong et al.~\cite{DBLP:conf/cocoon/HongELP12} gave a linear-time
testing and drawing algorithm to construct a straight-line 1-planar
drawing, if it exists. Recently, Nagamochi solved the more general
problem of straight-line drawability for wider classes of embedded
graphs~\cite{NagamochiTR}. On the other hand, Eggleton showed that
every outer-1-planar graph admits an outer-1-planar straight-line
drawing~\cite{eggleton}.

A drawn graph is called \emph{$k$-quasi planar} if it does not
contain $k$ mutually crossing edges. Fan-planar graphs are $3$-quasi
planar, since they cannot contain three independent edges that
mutually cross. It is conjectured that the number of edges of a
$k$-quasi planar graph is linear in the number of its vertices. Pach
et al.~\cite{DBLP:conf/jcdcg/PachRT02} and
Ackerman~\cite{DBLP:journals/dcg/Ackerman09} showed that this
conjecture holds for $3$- and $4$-quasi planar graphs, respectively.
Fox and Pach~\cite{DBLP:journals/siamdm/FoxPS13} showed that every
$k$-quasi-planar graph with $n$ vertices has at most
$O(n\log^{1+o(1)} n)$ edges.

A different forbidden crossing pattern arises in RAC
drawings where two edges are allowed to cross, as long as the crossings
edges form right angles. Graphs that admit such drawings (with
straight-line edges) are called \emph{right-angle crossing graphs}
or \emph{RAC graphs}, for short. Didimo et
al.~\cite{DBLP:journals/tcs/DidimoEL11} showed that a RAC graph with
$n$ vertices cannot have more than $4n - 10$ edges and that this
bound is tight. It is also known that a RAC graph is quasi
planar~\cite{DBLP:journals/tcs/DidimoEL11}, while a maximally dense
RAC graph (i.e., a RAC graph with $n$ vertices and exactly $4n - 10$
edges) is $1$-planar~\cite{DBLP:journals/dam/EadesL13}. Testing
whether a given graph is a RAC graph is
NP-hard~\cite{DBLP:journals/jgaa/ArgyriouBS12}. Dekhordi and
Eades~\cite{DBLP:journals/ijcga/DehkordiE12} proved that every
outer-1-plane graph has a straight-line RAC drawing, at the cost of
exponential area.

\subsection{Preliminaries}
\label{sec:preliminaries}
We consider finite, undirected and simple graphs. A \emph{drawing}
of a graph maps vertices to points in the plane and edges to simple
closed curves between the points corresponding to the end-vertices
of the edge. For a given drawing, we say that two edges \emph{cross}
if the interiors of their corresponding curves share a common point.
A drawing is \emph{simple} if no edge contains a vertex in its
interior, no two edges incident to a common end-vertex cross, no
edge crosses itself, no two edges meet tangentially, and no two
edges cross more than once. The \emph{rotation system} of a drawing
is the counterclockwise order of the incident edges around each
vertex. The \emph{embedding} of a drawn graph consists of its
rotation system and for each edge the sequence of edges crossing it.
For a graph $G$ and a vertex $v \in V[G]$, we denote by $G-\{v\}$
the graph that results from $G$ by removing $v$. A \emph{fan-planar}
drawing of a graph $G=(V,E)$ is a simple drawing such that, for each
edge $e \in E$, the edges that cross $e$, if any, are all incident
to a common vertex $v$ on the same side of $e$. Such a crossing is
called a \emph{fan-crossing}. An equivalent definition can be stated
by means of forbidden crossing patterns; see
Fig.~\ref{fig:fanplanar1}, \ref{fig:fanplanar2} and
\ref{fig:trianglecrossing}. A graph is \emph{fan-planar} if it
admits a fan-planar drawing. An drawing is \emph{outer-fan-planar}
if it is a fan-planar drawing with all vertices on the outer face. A
graph is \emph{outer-fan-planar} if it admits an outer-fan-planar
drawing. An outer-fan-planar graph is \emph{maximal
outer-fan-planar} if adding any edge to it yields a graph that is
not outer-fan-planar. Note that forbidden pattern II is irrelevant
for outer-fan-planarity.

We now briefly recall the SPQR-tree data
structure~\cite{gutwenger/mutzel:gd2000}. Two vertices $v$ and $w$
are a \emph{separation pair} of a connected graph $G$ if the graph
that results from $G$ by deleting $v$ and $w$ is not connected.  A
graph is \emph{3-connected} if it contains more than three vertices
but no separation pair.  An \emph{SPQR-tree} is a labeled tree that
represents the decomposition of a biconnected graph into 3-connected
components.  Each node $x$ of an SPQR-tree is labeled with a
multi-graph $G_x$ -- called the skeleton of $x$. There are four
different types of labels with the following skeletons:
\begin{inparaenum}[(i)]
\item \emph{$S$-nodes:} a simple cycle.
\item \emph{$P$-nodes:} three or more parallel edges.
\item \emph{$R$-nodes:} a simple 3-connected graph.
\item \emph{$Q$-nodes:} a single edge.
\end{inparaenum}
No two $S$-nodes, nor two $P$-nodes are adjacent in an SPQR-tree.
For each node $x$ of an SPQR-tree there is a one-to-one
correspondence of the edges of the skeleton of $G_x$ and the edges
incident to $x$. Further, let $\{x,y\}$ be an edge of an SPQR-tree
and let $e_x$ and $e_y$ be the edges of $G_x$ and $G_y$,
respectively, that are assigned to $\{x,y\}$. Then, $e_x$ and $e_y$
have the same end-vertices -- say $u$ and $v$. Moreover, let $x'$
and $y'$ be two nodes in different connected components of $T$
without the edge $\{x,y\}$. Then $G_{x'}$ and $G_{y'}$ share at most
$u$ and $v$ as common vertices.

An SPQR-tree represents the (multi-)graph constructed by iteratively
merging edges of the SPQR-tree as follows. For an edge $\{x,y\}$ of
the current tree, let $G_x$ and $G_y$ be the graphs currently
associated with $x$ and $y$, respectively. Remove the edge associated
with $\{x,y\}$ from both $G_x$ and $G_y$ -- except if they are
$Q$-nodes. Let the graph associated with the node that results from
merging $x$ and $y$ be the union of (the remaining parts of) $G_x$
and $G_y$.

The edges of a skeleton are called \emph{virtual edges} if they
correspond to a tree edge that is not incident to a $Q$-node and
\emph{real edges} otherwise. Note that real edges correspond to the
edges of the graph represented by the SPQR-tree.
Every biconnected graph has a unique SPQR-tree and the SPQR-tree
of a biconnected graph can be constructed in linear
time~\cite{gutwenger/mutzel:gd2000}.

\section{Recognizing and Drawing Maximal Outer-Fan-Planar Graphs}
\label{sec:outerfanplanar}

In this section, we prove that given a graph $G=(V,E)$ on $n$
vertices, there is a polynomial time algorithm to decide whether $G$
is maximal outer-fan-planar and if so a corresponding straight-line
drawing can be computed in linear time. We first observe that
biconnectivity is a necessary condition for maximal
outer-fan-planarity. In other words, a simply connected graph that
is not biconnected cannot be maximal outer-fan-planar. Indeed, if an
outer-fan-planar drawing has a cut-vertex $c$, it is easy to see
that it is always possible to draw an edge connecting two neighbors
of $c$ while preserving the outer-fan-planarity. Further, because of
the following lemma, we only have to check whether $G$ admits a
straight-line fan-planar drawing on a circle~$\mathcal{C}$; such a
drawing is completely determined by the cyclic ordering of the
vertices on $\mathcal{C}$.

\begin{lemma}
A biconnected graph $G$ is outer-fan-planar if and only if it admits a
straight-line outer-fan-planar drawing in which the vertices of $G$ are
restricted on a circle $\mathcal{C}$.
\label{lemma:outer_straight}
\end{lemma}
\begin{proof}
Let $G$ be an outer-fan-planar graph and let $\Gamma$ be an
outer-fan-planar drawing of $G$. We will only show that $G$ has a
straight-line outer-fan-planar drawing whose vertices lie on a
circle $\mathcal{C}$ (the other direction is trivial). The order of
the vertices along the outer face of $\Gamma$ completely determines
whether two edges cross, as in a simple drawing no two incident
edges can cross and any two edges can cross at most once. Now,
assume that two edges cross another edge in $\Gamma$. Then, both
edges have to be incident to the same vertex; hence, cannot cross
each other. So, the order of the crossings on an edge is also
determined by the order of the vertices on the outer face.
Therefore, we can construct a drawing $\Gamma_{\mathcal{C}}$ by
placing the vertices of $G$ on a circle $\mathcal{C}$ preserving
their order in the outer face of $\Gamma$ and draw the edges as
straight-line segments. \qed
\end{proof}

Since fan-planar graphs with $n$ vertices have at most $5n-10$
edges~\cite{DBLP:journals/corr/KaufmannU14}, we may assume that the
number of edges is linear in the number of vertices. We first
consider the case that $G$ is 3-connected (see
Section~\ref{subsec:tricon}) and then using SPQR-trees we show how
the problem can be solved for biconnected graphs (see
Section~\ref{sec:bicon}).

\subsection{The 3-Connected Case}
\label{subsec:tricon}

Assume that a straight-line drawing of a 3-connected graph $G$ with
$n$ vertices on a circle $\mathcal{C}$ is given. Let $v_1,\dots,v_n$
be the order of the vertices around $\mathcal{C}$. An edge
$\{v_i,v_j\}$ is an \emph{outer edge}, if $i-j\equiv \pm 1 \pmod n$,
a \emph{2-hop}, if $i-j\equiv \pm 2 \pmod n$, and a \emph{long edge}
otherwise. $G$ is a \emph{complete 2-hop graph}, if there are all
outer edges and all 2-hops, but no long edges. Two crossing long
edges are a \emph{scissor} if their end-points form two consecutive
pairs of vertices on $\mathcal{C}$. We say that a triangle is an
\emph{outer triangle} if two of its three edges are outer edges. We
call an outer-fan-planar drawing \emph{maximal}, if adding any edge
to it yields a drawing that is not outer-fan-planar.

Our algorithm is based on the observation that if a graph is
3-connected maximal outer-fan-planar, then it is a complete 2-hop
graph, or we can repeatedly remove any degree-3 vertex from any
4-clique until only a triangle is left. In a second step, we
reinsert the vertices maintaining outer-fan-planarity (if possible).
It turns out that we have to check a constant number of possible
embeddings. In the following, we prove some necessary properties.
The first three lemmas are used in the proof of
Lemma~\ref{LEMMA:maxk4}. Their proofs are based on the
3-connectivity of the input graph; see also
Fig.~\ref{fig:fan-planar-properties-1},
\ref{fig:fan-planar-properties-2} and
\ref{fig:fan-planar-properties-3}.

\begin{lemma}
Let $G$ be a 3-connected outer-fan-planar graph embedded on a circle
$\mathcal{C}$. If two long edges cross, then two of its end-points
are consecutive on $\mathcal{C}$. \label{LEMMA:cross_consec}
\end{lemma}
\begin{proof}
Assume to the contrary that there exist two long crossing edges
$\{v_{i_1},v_{i_3}\}$ and $\{v_{i_2},v_{i_4}\}$, such that $2 \leq
i_1 \leq i_2 - 2 \leq i_3 - 4 \leq i_4 - 6$ and $i_4 = n$; see
Fig.\ref{fig:fan-planar-properties-1}. Since $G$ is $3$-connected,
there has to be a vertex $v_{j_1}$ with $i_1 < j_1 < i_2$, such that
$v_{j_1}$ is adjacent to a vertex not in
$\{v_{i_1},\dots,v_{i_2}\}$. By outer-fan-planarity, this can only
be $v_{i_3}$ or $v_{i_4}$; say without loss of generality $v_{i_3}$.
Likewise there is a vertex $v_{i_2}$ with $i_2 < j_2 < i_3$, such
that $v_{j_2}$ is adjacent to a vertex not in
$\{v_{i_2},\dots,v_{i_3}\}$. By outer-fan-planarity this can now
only be $v_{i_4}$.
But now outer-fan-planarity does not permit
to add an edge connecting the two parts separated by $v_{i_3}$ and
$v_{i_4}$. \qed
\end{proof}

\begin{figure}[t]
  \centering
  \begin{minipage}[b]{.21\textwidth}
    \centering
    \subfloat[\label{fig:fan-planar-properties-1}{}]
    {\includegraphics[width=\textwidth,page=1]{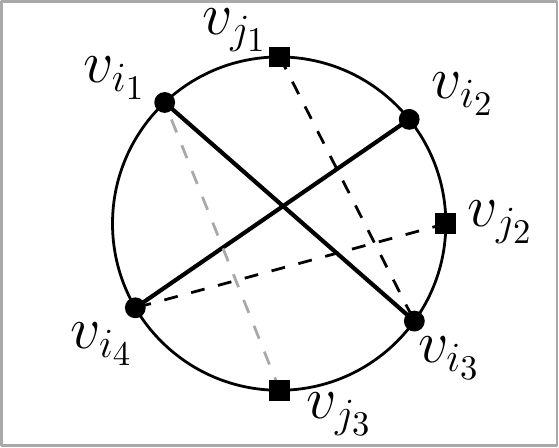}}
  \end{minipage}
  \begin{minipage}[b]{.21\textwidth}
    \centering
    \subfloat[\label{fig:fan-planar-properties-2}{}]
    {\includegraphics[width=\textwidth,page=2]{fan-planar-properties}}
  \end{minipage}
  \begin{minipage}[b]{.21\textwidth}
    \centering
    \subfloat[\label{fig:fan-planar-properties-3}{}]
    {\includegraphics[width=\textwidth,page=3]{fan-planar-properties}}
  \end{minipage}
  \begin{minipage}[b]{.21\textwidth}
    \centering
    \subfloat[\label{fig:fan-planar-properties-4}{}]
    {\includegraphics[width=\textwidth,page=4]{fan-planar-properties}}
  \end{minipage}
  \caption{Different configurations used in:
  (a)~Lemma~\ref{LEMMA:cross_consec},
  (b)~Lemma~\ref{LEMMA:scissor},
  (c, d)~Lemma~\ref{LEMMA:scissork4}.}
  \label{fig:fan-planar-properties-1}
\end{figure}

\begin{lemma}
Let $G$ be a 3-connected outer-fan-planar graph embedded on a circle
$\mathcal{C}$. If there are two long crossing edges, then there is a
scissor, as well. \label{LEMMA:scissor}
\end{lemma}
\begin{proof}
Let $e$ and $e'$ be two long crossing edges. By
Lemma~\ref{LEMMA:cross_consec}, it follows that two of the
end-points of $e$ and $e'$ are consecutive on $\mathcal{C}$. So,
assume without loss of generality that the vertices on $\mathcal{C}$
are labeled such that $e=\{v_1,v_k\}$ and $e'=\{v_\ell,v_n\}$ for
some $\ell < k$; see Fig.\ref{fig:fan-planar-properties-2}. If $k =
\ell + 1$, then the lemma holds. If this is not the case, then among
all crossing long edges with end-vertices $v_1$ and $v_n$ on one
hand and end-vertices between $v_\ell$ and $v_k$ on the other hand,
let edges $\{v_1,v_{j}\}$ and $\{v_n,v_i\}$, with $\ell \leq i < j
\leq k$ be the ones for which the difference $j-i$ is minimal.
Obviously, if $j=i+1$, then the edges $\{v_1,v_{j}\}$ and
$\{v_n,v_i\}$ are a scissor. Assume now that $j>i+1$. Since $v_i$
and $v_j$ cannot be a separation pair, there has to be an edge
between a vertex $v_{s}$ with $j < s < i$ and a vertex $v_t$ with $t
< i$ or $t > j$. By outer-fan-planarity $t = 1$ or $t = n$. This
contradicts the choice of $\{v_1,v_j\}$ and $\{v_i,v_n\}$. \qed
\end{proof}

\begin{lemma}
Let $G$ be a 3-connected graph  embedded on a circle $\mathcal{C}$
with a maximal outer-fan-planar drawing. If $G$ contains a scissor,
then its end-vertices induce a $K_4$. \label{LEMMA:scissork4}
\end{lemma}
\begin{proof}
Assume without loss of generality that the vertices on $\mathcal{C}$
are labeled such that $\{v_1,v_{i+1}\}$ and $\{v_i,v_n\}$ is a
scissor, for some $1 < i < n$. We have to show that $\{v_1,v_i\} \in
E[G]$ and $\{v_n,v_{i+1}\} \in E[G]$. By outer-fan-planarity there
cannot be an edge $\{v_\ell,v_k\}$, such that $1 < \ell < i$ and
${i+1} < k < n$; see Fig.~\ref{fig:fan-planar-properties-3}. Since
$v_1$ and $v_i$ cannot be a separation pair, there has to be an edge
between $v_n$ or $v_{i+1}$ and a vertex $v_\ell$ with $1 < \ell <
i$; say from $v_n$. Similarly, since $v_n$ and $v_{i+1}$ cannot be a
separation pair, there has to be an edge between $v_1$ or $v_i$ and
a vertex $v_k$, with ${i+1} < k < n$. By outer-fan-planarity, this
can only be an edge from $v_1$, as otherwise edge
$\{v_1,v_{i+1}\}$ would be crossed by two independent edges; see
Fig.~\ref{fig:fan-planar-properties-3}. As a consequence, there cannot be
an edge between $v_i$ and a vertex $v_k$, ${i+1} < k < n$ nor an
edge between $v_{i+1}$ and a vertex $v_\ell$, $1 < \ell < i$. Hence,
the edge $\{v_1,v_i\}$ is only crossed by edges incident to $v_n$.
Moreover, any edge that is crossed by $\{v_i,v_1\}$ is already
crossed by two edges incident to $v_1$. Since $G$ is maximal
outer-fan-planar, it must contain edge $\{v_i,v_1\}$. A similar
argument holds for $\{v_n,v_{i+1}\}$.\qed
\end{proof}

\begin{lemma}
Let $G$ be a 3-connected graph with a maximal outer-fan-planar
drawing and assume that the drawing contains at least one long edge.
Then, $G$ contains a $K_4$ with all four vertices drawn
consecutively on the circle. \label{LEMMA:maxk4}
\end{lemma}
\begin{proof}
First consider the case where the graph contains at least two
crossing long edges and, thus, by Lemma~\ref{LEMMA:scissor} a
scissor. Removing the vertices of a scissor, splits $G$ into two
connected components. Assume that we have chosen the scissor such
that the smaller of the two components is as small as possible
(thus, scissor-free) and that the vertices around $\mathcal{C}$ are
labeled such that this scissor is $\{v_1,v_{i+1}\}$, $\{v_i,v_n\}$
with $i \leq n - i$, i.e., the component induced by
$v_2,\dots,v_{i-1}$ is the smaller one. Recall that by
Lemma~\ref{LEMMA:scissork4} a scissor induces a $K_4$.

If $i=3$, i.e., if $\{v_1,v_3\}$ is a 2-hop, then $G$ should contain
either $\{v_2,v_n\}$ or $\{v_2,v_4\}$, as otherwise $v_1$ and $v_3$
is a separation pair; see Fig.~\ref{fig:fan-planar-properties-4}.
Say without loss of generality $\{v_2,v_n\}$. Then, $v_1$, $v_2$,
$v_3$ together with $v_n$ induce a $K_4$ with all vertices
consecutive on circle $\mathcal{C}$.

If $i>3$, let $\{v_k,v_\ell\}$, $1 \leq k < \ell \leq i$ be a long
edge such that there is no long edge
$\{v_{k'},v_{\ell'}\}\neq\{v_k,v_\ell\}$ with $k \leq k' < \ell'
\leq \ell$; see Fig.~\ref{fig:fan-planar-properties-5}. Then, no
long edge is crossing the edge $\{v_k,v_\ell\}$, as otherwise by
Lemma~\ref{LEMMA:scissor} such a crossing would yield a new scissor,
contradicting the choice of $\{v_1,v_{i+1}\}$ and $\{v_i,v_n\}$.
Since $\{v_k,v_\ell\}$ is not crossed by a long edge, it must be
crossed by exactly one 2-hop, say $\{v_{k-1},v_{k+1}\}$.  Now,
$\ell-k > 3$ is not possible, since we could add the edge
$\{v_{k+1}, v_\ell\}$, which is long. Hence, $\ell-k = 3$ and by
maximality of the outer-fan-planar drawing,
$v_k,v_{k+1},v_{k+2},v_\ell$ induces a $K_4$ with all vertices
consecutive on $\mathcal{C}$. Finally, if $G$ contains no two
crossing long edges, let $\{v_k,v_\ell\}$, $1 \leq k < \ell \leq n$
be a long edge such that there is no long edge
$\{v_{k'},v_{\ell'}\}\neq\{v_k,v_\ell\}$ with $k \leq k' < \ell'
\leq \ell$. By the same argumentation as above, we obtain that
$v_k,v_{k+1},v_{k+2},v_\ell$ induces a $K_4$ with all vertices
consecutive on $\mathcal{C}$. \qed
\end{proof}

\begin{lemma}\label{LEMMA:xyz1}
Let $G$ be a 3-connected outer-fan-planar graph with at least six
vertices. If $G$ contains a $K_4$ with all vertices drawn
consecutively on circle $\mathcal{C}$, then this $K_4$ contains
exactly one vertex of degree three and this vertex is neither the
first nor the last of the four vertices.
\end{lemma}
\begin{proof}
Let the vertices around circle $\mathcal{C}$ be labeled so that
$v_1$, $v_2$, $v_3$, $v_4$ induce a $K_4$. Since $v_1$ and $v_4$ is
not a separation pair, there is an edge between $v_2$ or $v_3$ and a
vertex, say $v_k$, among $v_5, \dots, v_n$. Hence, three out of the
four vertices $v_1$, $v_2$, $v_3$ and $v_4$ have degree at least
four; see Fig.~\ref{fig:fan-planar-properties-6}. If $v_3$ had a
neighbor in $v_5,\dots, v_n$, then this could only be $v_k$, as
otherwise $\{v_1,v_4\}$ would be crossed by two independent edges.
Since $G$ has at least 6 vertices, we assume without loss of
generality that $k > 5$. Since $v_4$ and $v_k$ is not a separation
pair, there has to be an edge $\{v_\ell,v_m\}$ for some $4 < \ell <
k$ and a $j \notin \{4,\dots,k\}$. But such an edge would not be
possible in an outer-fan-planar drawing.\qed
\end{proof}

\begin{figure}[t]
  \centering
  \begin{minipage}[b]{.21\textwidth}
    \centering
    \subfloat[\label{fig:fan-planar-properties-5}{}]
    {\includegraphics[width=\textwidth,page=5]{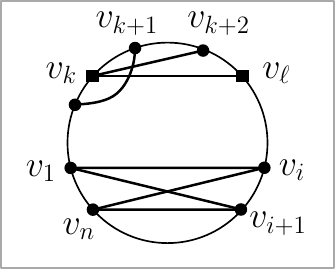}}
  \end{minipage}
  \begin{minipage}[b]{.21\textwidth}
    \centering
    \subfloat[\label{fig:fan-planar-properties-6}{}]
    {\includegraphics[width=\textwidth,page=6]{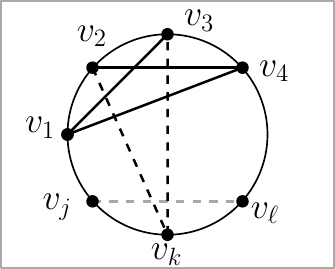}}
  \end{minipage}
  \begin{minipage}[b]{.21\textwidth}
    \centering
    \subfloat[\label{fig:fan-planar-properties-7}{}]
    {\includegraphics[width=\textwidth,page=7]{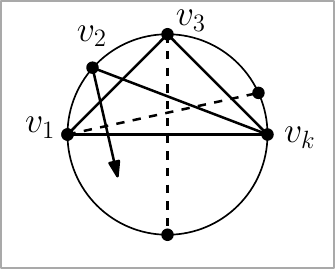}}
  \end{minipage}
  \begin{minipage}[b]{.21\textwidth}
    \centering
    \subfloat[\label{fig:fan-planar-properties-8}{}]
    {\includegraphics[width=\textwidth,page=8]{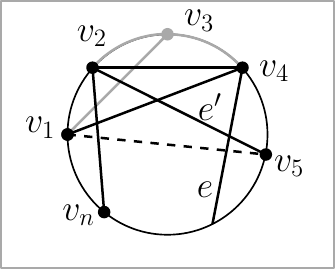}}
  \end{minipage}
  \caption{Different configurations used in:
  (a)~Lemma~\ref{LEMMA:maxk4},
  (b)~Lemma~\ref{LEMMA:xyz1},
  (c)~Lemma~\ref{LEMMA:xyz2},
  (d)~Lemma~\ref{LEMMA:G-v}.}
  \label{fig:fan-planar-properties-2}
\end{figure}

\begin{lemma}\label{LEMMA:xyz2}
Let $G$ be a 3-connected outer-fan-planar graph with at least six
vertices. If $G$ contains a $K_4$ with a vertex of degree 3, then
this $K_4$ has to be drawn consecutively on circle $\mathcal{C}$ in
any outer-fan-planar drawing of $G$.
\end{lemma}
\begin{proof}
Observe that any outer-fan-planar drawing of a $K_4$ contains
exactly one pair of crossing edges. If two 2-hops cross, then all
vertices of the $K_4$ are consecutive. If the $K_4$ contains two
crossing long edges, then each of the vertices of the $K_4$ is
incident to an outer edge not contained in the $K_4$; thus, has
degree at least four. If a long edge and a 2-hop cross, assume that
the vertices around $\mathcal{C}$ are labeled such that
$v_1,v_2,v_3,v_k$ induce a $K_4$ for some $5 \leq k < n$; see
Fig.~\ref{fig:fan-planar-properties-7}. Since $v_1$, $v_3$ and $v_k$
are incident to an outer edge not contained in the $K_4$, they have
degree at least four. We claim that $v_2$ has degree at least four.
Since $v_3$ and $v_k$ is not a separation pair, there is an edge
between a vertex among $v_4,\dots,v_{k-1}$ and $v_2$ or $v_1$ and an
edge between a vertex among $v_{k+1},\dots,v_{n}$ and $v_2$ or
$v_3$. Choosing $v_1$ and $v_3$ in the first and second case
respectively, yields two independent edges crossing $\{v_2,v_k\}$.
So, $v_2$ is connected to a vertex outside $K_4$. \qed
\end{proof}

\begin{lemma}
Let $G$ be a 3-connected graph with $n \geq 5$ vertices and let $v\in
V[G]$ be a vertex of degree three that is contained in a $K_4$.
Then, $G - \{v\}$ is 3-connected. \label{LEMMA:G-v-3con}
\end{lemma}
\begin{proof}
Let $a$, $b$, $c$ and $d$ be four arbitrary vertices of $G - \{v\}$.
Since $G$ was 3-connected, there was a path $P$ from $a$ to $b$ in
$G -\{c,d\}$. Assume that $P$ contains $v$. Since $v$ is only
connected to vertices that are connected to each other, there is
also another path from $a$ to $b$ in $G - \{c,d\}$ not containing
$v$. Hence, $a$ and $b$ cannot be a separation pair in $G - \{v\}$.
Since $a$ and $b$ were arbitrarily selected, $G-\{v\}$ is
3-connected.\qed
\end{proof}

\begin{lemma}
Let $G$ be a 3-connected graph with $n > 6$ vertices, let $v_1$,
$v_2$, $v_3$ and $v_4$ be four vertices that induce a $K_4$, such
that the degree of $v_3$ is three. Then, $G-\{v_3\}$ has a maximal
outer-fan-planar drawing if $G$ has a maximal outer-fan-planar
drawing. \label{LEMMA:G-v}
\end{lemma}
\begin{proof}
Consider a maximal outer-fan-planar drawing of $G$ on a circle
$\mathcal{C}$ and let $v_1$, $v_2$, $v_3$, $v_4,\dots,v_{n}$ be the
order of the vertices on $\mathcal{C}$ (recall
Lemma~\ref{LEMMA:xyz2}). Assume to the contrary that after removing
$v_3$, we could add an edge $e$ to the drawing; see
Fig.~\ref{fig:fan-planar-properties-8}. By Lemma~\ref{LEMMA:xyz1},
$\{v_3,v_1\}$ is the only edge incident to $v_3$ that crosses some
edges of $G-\{v_3\}$. Hence, there must be an edge $e'$ that is
crossed by $e$ and $\{v_3,v_1\}$. Since $\{v_3,v_1\}$ crosses only
edges incident to $v_2$ that also cross $\{v_1,v_4\}$, it follows
that $e'$ has to be incident to $v_2$. Further, since $G-\{v_3\}$
plus $e$ is outer-fan-planar it follows that $e$ is incident to
$v_1$ or $v_4$. Moreover, since $G$ plus $e$ is not outer-fan-planar
it follows that $e$ is incident to $v_4$.

Let $i$ be maximal so that there is an edge $\{v_2,v_i\}$. If $i
\neq n$, then $v_1$ and $v_i$ is a separation pair: Any edge
connecting $\{v_{i+1},\dots,v_{n-1}\}$ to
$\{v_2,v_3,\dots,v_{i-1}\}$ and not being incident to $v_2$ crosses
$\{v_2,v_i\}$. But edges crossing $\{v_2,v_i\}$ can only be incident
to $v_1$, a contradiction. Now, let $j > 4$ be minimum such that
there is an edge $\{v_2,v_j\}$. We claim that $j = 5$. If this is
not the case, then similarly to the previous case $v_4$ and $v_j$
would be a separation pair in $G-\{v_3\}$ plus $e$, which is not
possible due to Lemma~\ref{LEMMA:G-v-3con}.

It follows that $G$ has to contain edge $\{v_1,v_5\}$: Since $G$ is
outer-fan-planar, in $G$ there cannot be an edge $\{v_4,v_k\}$ for
some $k=6,\dots,n$, since it would cross $\{v_2,v_5\}$ which is
crossed by $\{v_3,v_1\}$. So, $\{v_1,v_5\}$ crosses only edges
incident to $v_2$ that are already crossed by $\{v_3,v_1\}$ and
$\{v_4,v_1\}$. Hence, $\{v_1,v_5\}$ could be added to $G$ without
violating outer-fan-planarity; a clear contradiction. Since $e$ and
$\{v_2,v_{n}\}$ both cross $\{v_1,v_5\}$ it follows that
$e=\{v_4,v_{n}\}$. But now, $v_5$ and $v_{n}$ has to be a separation
pair.\qed
\end{proof}

\begin{remark}\label{REM:6}
Let $G$ be a graph with 6 vertices containing a vertex $v$ of degree
three.  Then, $G$ is maximal outer-fan-planar if and only if
$G-\{v\}$ is a $K_5$ missing one of the edges that connects a
neighbor of $v$ to one of the other two vertices.
\end{remark}

\begin{lemma}\label{LEMMA:2-hop}
It can be tested in linear time whether a graph is a complete 2-hop
graph. Moreover, if a graph is a complete 2-hop graph, then it has a
constant number of outer-fan-planar embeddings and these can be
constructed in linear time.
\end{lemma}
\begin{proof}
Let $G$ be an $n$-vertex graph. We test whether $G$ is a complete
2-hop as follows. If $n\in \{4,5\}$, then $G$ is either $K_4$ or
$K_5$. Otherwise, check first whether all vertices have degree four.
If so, pick one vertex as $v_1$, choose a neighbor as $v_2$ and a
common neighbor of $v_1$ and $v_2$ as $v_3$ (if no such common
neighbor exists, then $G$ is not a complete 2-hop). Assume now that
we have already fixed $v_1,\dots,v_i$, $3 \leq i < n$. Test whether
there is a unique vertex $v \in V \setminus \{v_1,\dots,v_i\}$ that
is adjacent to $v_i$ and $v_{i-1}$. If so, set $v_{i+1}=v$.
Otherwise reject. If we have fixed the order of all vertices check
whether there are only outer edges and 2-hops. Do this for any
possible choices of $v_2$ and $v_3$, i.e., for totally at most 6
choices.\qed
\end{proof}

\begin{remark}\label{REM:2-hop+}
No degree 3 vertex can be added to an $n$-vertex complete 2-hop with
$n \geq 5$.
\end{remark}

We are now ready to describe our algorithm. If the graph is not a
complete 2-hop graph, recursively try to remove a vertex of degree 3
which is contained in a $K_4$. If $G$ is maximal outer-fan-planar,
Lemmas~\ref{LEMMA:maxk4} and~\ref{LEMMA:xyz1} guarantee that such a
vertex always exists in the beginning. Remark~\ref{REM:2-hop+}
guarantees that also in subsequent steps there is a long edge and,
thus, Lemmas~\ref{LEMMA:G-v-3con} and~\ref{LEMMA:G-v} guarantee that
also in subsequent steps, we can apply Lemma~\ref{LEMMA:maxk4} as
long as we have at least six vertices. Remark~\ref{REM:6} guarantees
that we can also remove two more vertices of degree 3 ending with a
triangle.

At this stage, we already know that if the graph is
outer-fan-planar, it is indeed maximal outer-fan-planar. Either, we
started with a complete 2-hop graph or we iteratively removed
vertices of degree three yielding a triangle. Note that in the
latter case we must have started with $3n-6$ edges. On the other
hand, if we apply the above procedure to an $n$-vertex 3-connected
maximal outer-fan-planar graph, we get that the number of edges is
exactly $2n$ or $3n-6$.

Finally, we try to reinsert the
vertices in the reversed order in which we have deleted them. By
Lemma~\ref{LEMMA:xyz2}, we can insert the vertex of degree three only
between its neighbor, i.e., there are at most two possibilities where
we could insert the vertex. Lemma~\ref{LEMMA:branch} guarantees that
in total, we have to check at most four possible drawings for $G$.
A summary of our approach is also given in Algorithm~\ref{ALGO:3con}.

\begin{algorithm}[p]
  \Input{3-connected graph $G=(V,E)$, subset $E' \subseteq E$}
  \Output{\textsc{true} if and only if $G$ is maximal outer-fan-planar and\\
    has an outer-fan-planar drawing in which edges in $E'$ are outer edges and if so\\
     all outer-fan-planar drawings of $G$ in which edges in $E'$ are outer edges}

  \BlankLine

  \Begin{
    mark all edges in $E'$\;
    \If{$G$ is a complete 2-hop graph}{
      \Return whether $G$ has an outer-fan-planar drawing with
      marked edges on outer face\;
    }
    \While{there are vertices of degree 3 contained in a $K_4$}{
      let $v$ be such a vertex, with neighbors $a,b,c$\;
        \If{$v$ is contained in three marked triangles or three marked edges}{
          \Return \textsc{false}\;
        }
        \ForAll{marked triangles $v,x,y$ containing $v$}{
          \lIf{the edge $\{x,y\}$ was not marked}{
            mark the edge $\{x,y\}$ with $v$
          }
        }
        $S$.\textsc{Push}($v$)\;
        remove $v$ from $G$\;
        mark the triangle $a,b,c$\;
    }
    \lIf{the remainder is not a triangle}{
      \Return \textsc{false}
    }
    \While{$S \neq \emptyset$}{
      $v \gets S.\textsc{Pop}$\;
      remove the mark from all edges marked $v$\;
      insert $v$ between two of its neighbors such that\\
      \hspace{0.2cm} $(a)$ all marked edges are outer edges and
      \hspace{0.2cm} $(b)$ outer-fan-planarity is preserved\;
      \lIf{both possibilities work}{branch}
      \lIf{no possibility works}{\Return \textsc{false}}
    }
    \Return \textsc{true}\;
  }
  \caption{\label{ALGO:3con}3-Connected Maximal Outer-Fan-Planarity}
\end{algorithm}

\begin{lemma}\label{LEMMA:branch}
When reinserting a sequence of degree 3 vertices starting from a
triangle, at most the first two vertices have two choices where they
could be inserted.
\end{lemma}
\begin{proof}
Let $H$ be a outer-fan-planar graph and let three consecutive
vertices $v_1,v_2,v_3$ induce a triangle. Assume, we want to insert
a vertex $v$ adjacent to $v_1,v_2,v_3$. By Lemma~\ref{LEMMA:xyz1},
we have to insert $v$ between $v_1$ and $v_2$ or between $v_2$ and
$v_3$. Note that the edges that are incident to $v_2$ and cross
$\{v_1,v_3\}$ are also crossed by an edge $e$ incident to $v$. So,
if there is an edge incident to $v_2$ that was already crossed twice
before inserting $v$, this would uniquely determine whether $e$ is
incident to $v_1$ or $v_3$ and, thus, where to insert $v$.

We will now show that after the first insertion each relevant vertex
is incident to an edge that is crossed at least twice. When we
insert the first vertex we create a $K_4$. From the second vertex
on, whenever we insert a new vertex, it is incident to an edge that
is crossed at least twice. Also, after inserting the second degree 3
vertex, three among the four vertices of the initial $K_4$ are also
incident to an edge that is crossed at least twice. The forth vertex
of the initial $K_4$ is not the middle vertex of a triangle
consisting of three consecutive vertices. It can only become such a
vertex if its incident inner edges are crossed by a 2-hop. But then
these inner edges are all crossed at least twice. \qed
\end{proof}

Summarizing, we obtain the following theorem; in order to exploit
this result in the biconnected case, it is also tested whether a
prescribed subset (possibly empty) of edges can be drawn as outer
edges.

\begin{theorem}
Given a 3-connected graph $G$ with a subset $E'$ of its edge set, it
can be tested in linear time whether $G$ is maximal outer-fan-planar
and has an outer-fan-planar drawing such that the edges in $E'$ are
outer edges. Moreover if such a drawing exists, it can be
constructed in linear time.
\end{theorem}
\begin{proof}
Let $n$ be the number of vertices. By Lemma~\ref{LEMMA:2-hop}, a
complete 2-hop graph has only a constant number of outer-fan-planar
embeddings which can be computed in linear time. In the other case,
any vertex that was removed from the queue will never be appended
again. Hence, there are at most $n$ iterations in the first part of
Algorithm~\ref{ALGO:3con}.

To check whether the degree three vertices can be reinserted back in
the graph, we only have to consider in total four different
embeddings. Assume that we want to insert a vertex $v$ into an outer
triangle $v_1,v_2,v_3$. Then we just have to check whether $v_1$ or
$v_3$ are incident to edges other than the edge $\{v_1,v_3\}$ that
cross an edge incident to $v_2$. This can be done in constant time
by checking only two pairs of edges.\qed
\end{proof}

\subsection{The Biconnected Case}
\label{sec:bicon}

We now show how to test outer-fan-planar maximality on a biconnected graph.

\begin{lemma}\label{LEMMA:porous}
Let $v_1,\dots,v_n$ be the order of the vertices around the circle
in an outer-fan-planar drawing of a 3-connected graph $G$. If we can
add a vertex $v$ between $v_1$ and $v_n$ with an edge $\{v,v_i\}$
for some $i=2,\dots,n-1$, then $i=2$ or $i=n-1$.
\end{lemma}
\begin{proof}
Otherwise, since $v_1,v_i$ cannot be a separation pair of $G$, there
has to be an edge from a $v_k$ for some $k=2,\dots,i-1$ that crosses
$\{v,v_i\}$ and hence an edge $\{v_k,v_n\}$. Since $v_n,v_i$ cannot
be a separation pair of $G$, there has to be an edge
$\{v_1,v_\ell\}$ for some $\ell=i+1,\dots,n-1$. But now there are
three independent edges crossing.\qed
\end{proof}

We say that an outer edge $\{v_1,v_n\}$ is \emph{porous} around
$v_1$ if we could add a vertex $v$ between $v_1$ and $v_n$ and an
edge $\{v,v_2\}$ maintaining outer-fan-planarity. Note that any edge
of a simple cycle, i.e., of the skeleton of an $S$-node is porous
around any of its end-vertices.  Any outer edge of a $K_4$ is porous
around any of its end-vertices; see Fig.~\ref{fig:porous}.

\begin{figure}[t]
  \centering
  \begin{minipage}[b]{.28\textwidth}
    \centering
    \subfloat[\label{fig:porous1}{}]
    {\includegraphics[width=\textwidth,page=1]{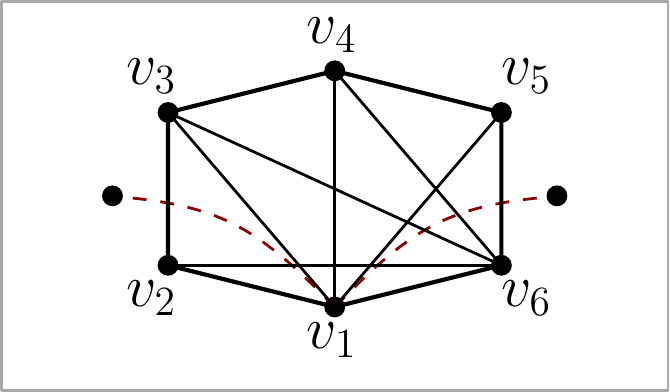}}
  \end{minipage}
  \begin{minipage}[b]{.28\textwidth}
    \centering
    \subfloat[\label{fig:porous2}{}]
    {\includegraphics[width=\textwidth,page=2]{porous}}
  \end{minipage}
  \caption{%
  (a)~In the solid graph, edge $\{v_2,v_3\}$ ($\{v_5,v_6\}$) is porous around $v_2$ ($v_6$, resp.).
  (b)~Illustration of Case~\ref{CASE:porous2} of Theorem~\ref{THEO:char}.}
  \label{fig:porous}
\end{figure}

We use the SPQR-tree of a biconnected graph to characterize whether
it is maximal outer-fan-planar.

\begin{theorem}
A biconnected graph is maximal outer-fan-planar iff the following
hold:
\begin{enumerate}[1)]
\item \label{CASE:Rnode} The skeleton of any $R$-node is maximal
outer-fan-planar and has an outer-fan-planar drawing in which all
virtual edges are outer edges,
\item \label{CASE:RS} No $R$-node is adjacent to an $R$-node or an $S$-node,
\item \label{CASE:Snode} All $S$-nodes have degree three,
\item \label{CASE:Pnode} All $P$-nodes have degree three and are adjacent to a
$Q$-node, and
\item \label{CASE:porous} Let $G_1$ and $G_2$ be the skeleton of the two neighbors of a
$P$-node other than the $Q$-node and let $\{s,t\}$ be the common
virtual edge of $G_1$ and $G_2$. Then, $G_i,i=1,2$ must not admit an
outer-fan-planar drawing with $t_i,s,t,s_i$ being consecutive around
the circle and
\begin{enumerate}[(a)]
\item \label{CASE:porous1} edge $\{s,t\}$ is porous in both $G_1$ and $G_2$ around the same
vertex, or
\item \label{CASE:porous2} edge $\{t_1,s\}$ ($\{s_2,t\}$) is real and porous around
$s$ ($t$, resp.), or
\item edge $\{s_1,t\}$ ($\{t_2,s\}$ ) is real and porous around $t$ ($s$, resp.).
\end{enumerate}
\end{enumerate}
\label{THEO:char}
\end{theorem}
\begin{proof}
Let $G$ be a biconnected graph.
\begin{description}
\item[$\mathbf \Leftarrow$:] Clearly, if \ref{CASE:Rnode} and
\ref{CASE:Pnode} are fulfilled, then $G$ is outer-fan-planar. Just
merge skeletons at common virtual edges such that one skeleton is in
the outer face of the other skeleton. It remains to show maximality.

The skeleton of each node is maximal outer-fan-planar. Assume now
that we have already merged some nodes of the SPQR-tree obtaining
a maximal outer-fan-planar graph $H$ and that we next want to
merge $H$ with a skeleton $G_x$ at a virtual edge $\{s,t\}$
obtaining a graph $H'$. There is nothing to show if $G_x$ is the
skeleton of a $P$-node. So assume that $G_x$ is a triangle or a
3-connected graph. Consider a fixed outer-fan-planar drawing of
$H'$.

We first show that the vertices of $G_x$ are consecutive on the
circle.  If $G_x$ is a triangle, this follows directly from
Lemma~\ref{LEMMA:porous}, Condition \ref{CASE:porous2}, and its
symmetric counter part. Assume now that $G_x$ is 3-connected. Note
that the outer-fan-planar drawing of $H'$ also induce
outer-fan-planar drawings of $G_x$ and $H$. By maximality these two
drawings contain all outer edges. Let now $e$ be an edge that was an
outer edge in one of the two subgraphs~--~say $G_x$~--~but not
incident to $s$ or $t$. Then, $e$ can only be an outer edge or a
2-hop in $H'$: if $e'$ is not an outer edge in $H$, it crosses at
least two outer edges of the other subgraph~--~here $H$. Hence,
starting from $s$ and $t$ the vertices must be ordered as follows
around the circle: first there might be one or more vertices of one
of the two subgraphs $G_x$ or $H$. Then, there might be
alternatingly a vertex from $H$ and $G_x$, finally there could be
again several vertices from one of the two subgraphs $G_x$ or $H$.
Using that $G_x$ is 3-connected and a case distinction on whether
the two sequences of vertices next to $s$ and $t$, respectively, are
chosen from $H$ or $G_x$, respectively, we obtain that the
alternating part on the circle has to be empty.

Hence, $G_x$ has to be inserted next to $s$ or $t$. Now, Condition
\ref{CASE:porous2} implies that $G_x$ must be inserted right between
$s$ and $t$. Hence, the only edge that could be inserted into the
drawing would be an edge crossing $\{s,t\}$ which is prohibited by
Condition \ref{CASE:porous2} and its symmetric counter part.
\item[$\mathbf \Rightarrow$:] Assume again that we have already merged
  some nodes of the SPQR-tree obtaining a maximal outer-fan-planar
  graph $H$ and that we next want to merge $H$ with a skeleton $G_x$
  at a virtual edge $\{s,t\}$ obtaining a graph $H'$. Note that in a
  maximal outer-fan-planar drawing all outer edges have to be
  present. This implies especially, that the vertices of $G_x$ have to
  be consecutive in an outer-fan-planar drawing of $H'$ with $s$ and
  $t$ being the first and the last vertex  and that we cannot draw the
  skeletons of two nodes adjacent to one $P$-node on the same side of
  the respective virtual edge. Otherwise, we obtain the situation
  indicated in Fig.~\ref{fig:porous2}. This implies \ref{CASE:Rnode},
  \ref{CASE:porous2}, and its symmetric counterpart.  Moreover, all
  virtual edges have to be real edges which implies \ref{CASE:RS} and,
  combined with the previous observation, also \ref{CASE:Pnode}. If
  the skeleton of an $S$-node would be a cycle of length greater than
  three, we could add chords, contradicting maximality. Hence,
  \ref{CASE:Snode} is fulfilled. Finally, when combining two
  components in parallel, we should not be able to add an edge from
  one component to the other routed over the virtual edge. This
  implies \ref{CASE:porous1}.\qed
\end{description}
\end{proof}

\section{The NP-hardness of the \fpfrs{} Problem}
\label{sec:NPhard}

In this section, we study the \fpfrs{} problem (FP-FRS), that is,
the problem of deciding whether a graph $G = (V, E)$ with a fixed
rotation system $\mathcal{R}$ admits a fan-planar drawing preserving
$\mathcal{R}$.

\begin{theorem}
\fpfrs{} is NP-hard.
\end{theorem}
\begin{proof}

We prove the statement by using a reduction from \threepart{} (3P).
An instance of 3P is a multi-set $A = \{a_1, a_2, \ldots, a_{3m} \}$
of $3m$ positive integers in the range $(B/4, B/2)$, where $B$ is an
integer such that $\sum_{i=1}^{3m} a_i = mB$. 3P asks whether $A$
can be partitioned into $m$ subsets $A_1, A_2, \ldots, A_m$, each of
cardinality $3$, such that the sum of the numbers in each subset is
$B$. As 3P is \emph{strongly} NP-hard~\cite{gj-cigtnpc-79}, it is
not restrictive to assume that $B$ is bounded by a polynomial in
$m$.

Given an instance $A$ of 3P, we show how to transform it into an
instance $\langle G_A, \mathcal{R}_A \rangle$ of FP-FRS, by a
polynomial-time transformation, in such a way that the former is a
\emph{Yes}-instance of 3P if and only if the latter is a
\emph{Yes}-instance of FP-FRS.

Before describing our transformation in detail, we need to introduce
the concept of \emph{barrier gadget}. An $n$-vertex \emph{barrier
gadget} is a graph consisting of a cycle of $n \geq 5$ vertices plus
all its $2$-hop edges; a barrier gadget is therefore a maximal
outer-$2$-planar graph. We make use of barrier gadgets in order to
constraint the routes of some specific paths of $G_A$, as will be
clarified soon. We exploit the following property of barrier
gadgets. Let $G$ be a biconnected fan-planar graph containing a
barrier subgraph $G_b$, and let $\Gamma$ be a fan-planar drawing of
$G$ such that drawing $\Gamma_b$ of $G_b$ in $\Gamma$ is maximal
outer-$2$-planar. Then, no path $\pi$ of $G - G_b$ can enter inside
the boundary cycle of $\Gamma_b$ and cross a $2$-hop edge. Indeed,
every $2$-hop edge $e_b$ of $\Gamma_b$ is crossed by two other
$2$-hop edges having an end-vertex in common, hence if $e_b$ were
crossed by $\pi$, then $e_b$ would be crossed by two independent
edges. On the other hand, if path $\pi$ enters inside $\Gamma_b$
without crossing any $2$-hop edge, then it must cross twice a same
boundary edge $e_b'$ because of the biconnectivity of $G$; namely,
if path $\pi$ enters in $\Gamma_b$, then it must also exit from it
passing through the same boundary edge. In this case, the only
possibility that preserves the fan-planarity of $\Gamma$ is that
$\pi$ crosses $e_b'$ with two consecutive edges, thus forming a
fan-crossing. Otherwise, $e_b'$ would be crossed either by two
independent edges of $\pi$ or by a same edge of $\pi$ twice, but
both these cases are not allowed in a (simple) fan-planar drawing.

Now, we are ready to describe how to transform an instance $A$ of 3P
into an instance $\langle G_A, \mathcal{R}_A \rangle$ of FP-FRS. We
start from the construction of graph $G_A$ which will be always
biconnected. First of all, we create a \emph{global ring barrier} by
attaching four barrier gadgets $G_t$, $G_r$, $G_b$ and $G_l$ as
depicted in Figure~\ref{fi:3P-reduction}. $G_t$ is called the
\emph{top beam} and contains exactly $3mK$ vertices, where $K =
\lceil B/2 \rceil + 1$. $G_r$ is the \emph{right wall} and has only
five vertices. $G_b$ and $G_r$ are called the \emph{bottom beam} and
the \emph{left wall}, respectively, and they are defined in a
specular way. Observe that $G_t$, $G_r$, $G_b$ and $G_l$ can be
embedded so that all their vertices are linkable to points within
the closed region delimited by the global ring barrier. Then, we
connect the top and bottom beams by a set of $3m$ \emph{columns},
see Figure~\ref{fi:3P-reduction} for an illustration of the case $m
= 3$. Each \emph{column} consists of a stack of $2m-1$ \emph{cells};
a \emph{cell} consists of a set of pairwise disjoint edges, called
the \emph{vertical edges} of that cell. In particular, there are
$m-1$ \emph{bottommost cells}, one \emph{central cell} and $m-1$
\emph{topmost cells}. Cells of a same column are separated by $2m-2$
barrier gadgets, called \emph{floors}. Central cells (that are $3m$
in total) have a number of vertical edges depending on the elements
of $A$. Precisely, the central cell $C_i$ of the $i$-th column
contains $a_i$ vertical edges connecting its delimiting floors ($i
\in \{1,2, ..., 3m\}$). Instead, all the remaining cells have, each
one, $K$ vertical edges. Hence, a non-central cell contains more
edges than any central cell. Further, the number of vertices of a
floor is given by the number of its incident vertical edges minus
two. Let $u$ and $v$ be the ``central'' vertices of the left and
right walls, respectively (see also Figure~\ref{fi:3P-reduction}).
We conclude the construction of graph $G_A$ by connecting vertices
$u$ and $v$ with $m$ pairwise internally disjoint paths, called the
\emph{transversal paths} of $G_A$; each transversal path has exactly
$(3m-3)K+B$ edges.

Concerning the choice of a rotation system $\mathcal{R}_A$, we
define a cyclic ordering of edges around each vertex that is
compatible with the following constraints: (i) every barrier gadget
can be embedded with all its $2$-hop edges inside its boundary
cycle; (ii) the global ring barrier can be embedded with only four
vertices on the outer face; (iii) columns can be embedded inside the
region delimited by the global ring barrier without crossing each
other; (iv) vertical edges of cells can be embedded without creating
crossings; (v) transversal paths are attached to the left and right
walls such that the ordering of their edges around $u$ is specular
to the ordering around $v$; this choice makes it possible to avoid
crossings between any two transversal paths. From what said, it is
straightforward to see that an instance of 3P can be transformed
into an instance of FP-FRS in polynomial time in $m$.

\begin{figure}[t!]
    \centering
    \includegraphics[page=1, width=.85\textwidth]{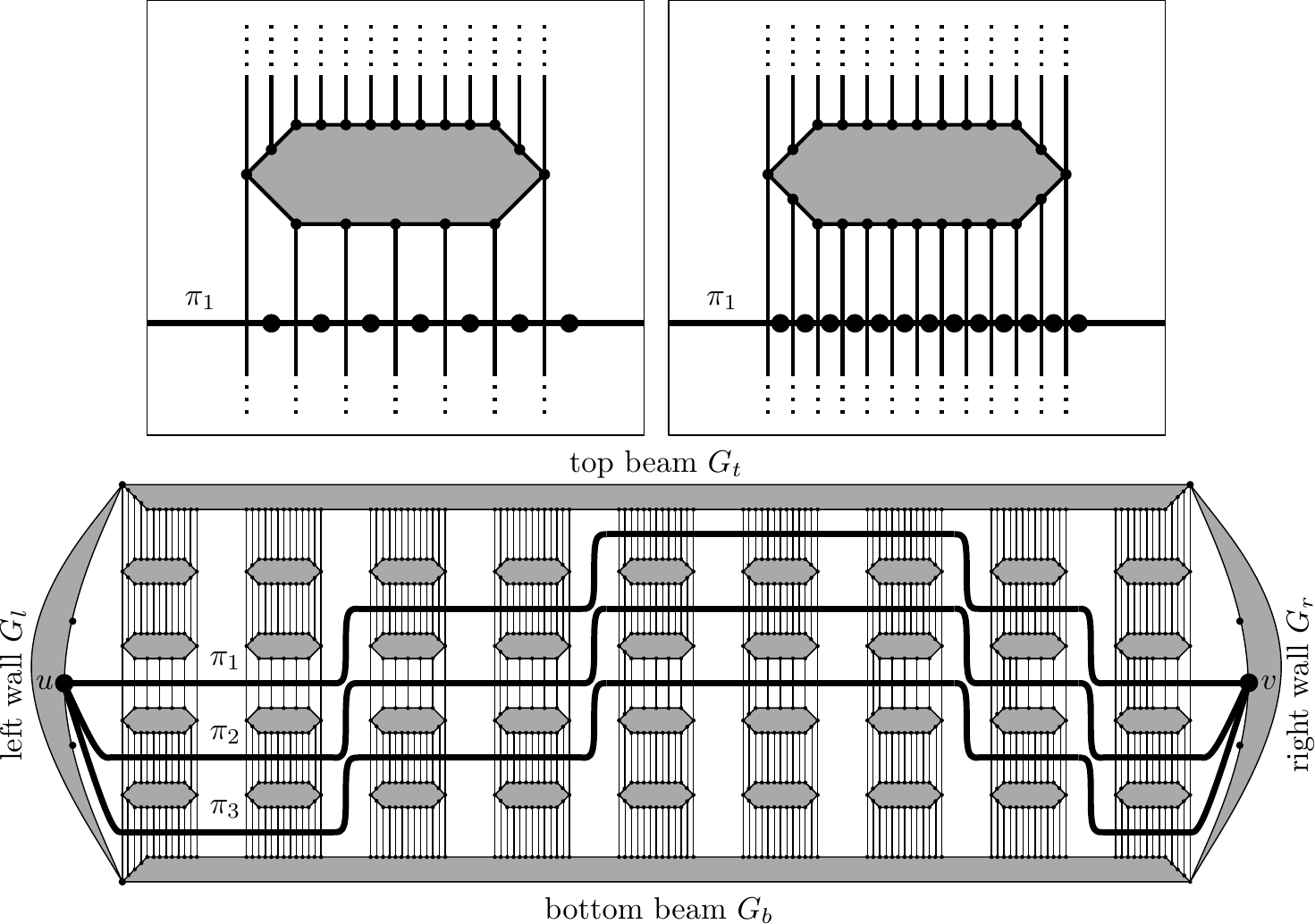}
    \caption{Illustration of the reduction of FP-FRS from 3P, where $m=3$, $A=\{7,7,7,8,8,8,8,9,10\}$ and $B = 24$. Transversal paths are routed according to the following solution of 3P: $A_1=\{7,7,10\}$, $A_2=\{7,8,9\}$ and $A_3=\{8,8,8\}$. The top-left and top-right boxes show a zoom of the first central cell and of the non-central cell of the $3$-rd column traversed by path $\pi_1$, respectively.}
    \label{fi:3P-reduction}
\end{figure}

We now prove that a \emph{Yes}-instance of 3P is transformed into a
\emph{Yes}-instance of FP-FRS, and vice-versa. Let $A$ be a
\emph{Yes}-instance of 3P, we show that $\langle G_A, \mathcal{R}_A
\rangle$ admits a fan-planar drawing $\Gamma_A$ preserving
$\mathcal{R_A}$. We preliminarily observe that such a drawing is
easy to compute if one omits all the transversal paths. It is
essentially a drawing like that one depicted in
Figure~\ref{fi:3P-reduction}, where columns are one next to the
other within the closed region delimited by the global ring barrier.
However, by exploiting a solution $\{A_1, A_2, \ldots, A_m\}$ of 3P
for the instance $A$, also the transversal paths can be easily
embedded without violating the fan-planarity. The idea is to route
these paths in such a way that: (R.1) they do not cross each other;
(R.2) they do not cross any barrier; (R.3) each path passes through
exactly $3$ central cells and $3m-3$ non-central cells; and (R.4)
each cell is traversed by at most one path. More precisely, each
transversal path $\pi_j$ is biunivocally associated with a subset
$A_j$ ($j \in \{1, 2, \ldots, m\}$) and the three central cells it
passes through have three sets of vertical edges whose cardinalities
form a triple of integers matching the three integers of $A_j$.
Paths are routed by sweeping columns from left to right and by
proceeding as follows. Let $C_1$ be the $1$-st central cell; $C_1$
has $a_1$ vertical edges by construction. The transversal path
passing through $C_1$ is a path $\pi_{j(1)}$ such that subset
$A_{j(1)}$ contains an integer equal to $a_1$. The remaining
transversal paths are routed until the $1$-st column by preserving
the cyclic edge-ordering around $u$ and by respecting conditions
R.1, R.2, R.3 and R.4; note that condition R.3 cannot be violated at
this point. Suppose now that all paths have been already routed
until the $(i-1)$-th column, for some $i \geq 2$, and suppose also
that conditions R.k ($1 \leq  k \leq 4$) are satisfied. Then, there
is at least a path $\pi_{j(i)}$ whose corresponding subset
$A_{j(i)}$ contains an integer $a_i$ that has not yet been
considered. Path $\pi_{j(i)}$ is the next path that goes through the
central cell $C_i$. The remaining paths are routed in such a way
that their ``vertical distances'' to path $\pi_{j(i)}$, in terms of
number of cells, are unchanged when passing from the $(i-1)$-th
column to the $i$-th column. Eventually, each transversal path
crosses exactly $(3m-3)K+B$ vertical edges, which is the same number
of its edges. Therefore, it is possible to draw these paths by
ensuring that each of their edges crosses exactly one vertical edge,
which preserves the fan-planarity. Hence, eventually we get a
fan-planar drawing $\Gamma_A$ preserving the rotation system
$\mathcal{R}_A$.

We conclude the proof by showing that if $\langle G_A, \mathcal{R}_A
\rangle$ is a \emph{Yes}-instance of FP-FRS, then $A$ is a
\emph{Yes}-instance of 3P. Let $\Gamma_A$ be a fan-planar drawing of
$G_A$ preserving the rotation system $\mathcal{R}_A$. We first
observe that the top beam and the bottom beam are disjoint,
otherwise there would be at least a 2-hope edge in one beam that is
crossed by another edge of the other beam, thus violating the
fan-planarity. We also note that columns can partially cross each
other, but this does not actually affect the validity of the proof.
Indeed, an edge $e$ of a column $L$ might cross an edge $e'$ of
another column $L'$ only if $e$ is incident to a vertex in the
rightmost (leftmost) side of $L$, $e'$ is a leftmost (rightmost)
vertical edge of $L'$, and $L$ and $L'$ are two consecutive columns.
With a similar argument, it is immediate to see that vertices $u$
and $v$ must be separated by all the columns. Therefore, every
transversal path satisfies conditions R.1, R.2 and it must pass
through at least three central cells, if not it would cross a number
of pairwise disjoint edges that is greater than the number of its
edges, hence $\Gamma_A$ would not be fan-planar. On the other hand,
because of condition R.4, which is obviously satisfied, there cannot
be any transversal path passing through more than three central
cells. Otherwise, there would be some other transversal path that
traverses a number of central cells that is strictly less than
three. Hence, also condition R.3 is satisfied. In conclusion, every
transversal path $\pi_j$ ($j \in \{1, 2, \ldots, m\}$) crosses
$(3m-3)K+B$ vertical edges and traverses exactly three central cells
$C_{1j}$, $C_{2j}$ and $C_{3j}$. If $m(C_{1j}), m(C_{2j})$ and
$m(C_{3j})$ denote the number of edges of these cells, then
$m(C_{1j}) + m(C_{2j}) + m(C_{3j}) = B$, because each non-central
cell has $K$ edges. Therefore, the partitioning of $A$ defined by
$A_1, A_2, \ldots, A_m$, where $A_j = \{m(C_{1j}), m(C_{2j}),
m(C_{3j})\}$, is a solution of 3P for the instance $A$. \qed
\end{proof}

\bibliographystyle{splncs03}
\bibliography{references}

\begin{thebibliography}{10}
\providecommand{\url}[1]{\texttt{#1}}
\providecommand{\urlprefix}{URL }

\bibitem{DBLP:journals/dcg/Ackerman09}
Ackerman, E.: On the maximum number of edges in topological graphs with no four
  pairwise crossing edges. Discrete {\&} Computational Geometry  41(3),
  365--375 (2009)

\bibitem{DBLP:journals/combinatorica/AgarwalAPPS97}
Agarwal, P.K., Aronov, B., Pach, J., Pollack, R., Sharir, M.: Quasi-planar
  graphs have a linear number of edges. Combinatorica  17(1),  1--9 (1997)

\bibitem{DBLP:journals/jgaa/ArgyriouBS12}
Argyriou, E.N., Bekos, M.A., Symvonis, A.: The straight-line {RAC} drawing
  problem is np-hard. J. Graph Algorithms Appl.  16(2),  569--597 (2012)

\bibitem{DBLP:conf/gd/AuerBBGHNR13}
Auer, C., Bachmaier, C., Brandenburg, F.J., Glei{\ss}ner, A., Hanauer, K.,
  Neuwirth, D., Reislhuber, J.: Recognizing outer 1-planar graphs in linear
  time. In: Wismath, S., Wolff, A. (eds.) GD 2013. LNCS, vol. 8242, pp.
  107--118. Springer, Heidelberg (2013)

\bibitem{MANA:MANA3211170125}
Bodendiek, R., Schumacher, H., Wagner, K.: {\"U}ber 1-optimale graphen.
  Mathematische Nachrichten  117(1),  323--339 (1984)

\bibitem{DBLP:journals/corr/abs-1203-5944}
Cabello, S., Mohar, B.: Adding one edge to planar graphs makes crossing number
  and 1-planarity hard. CoRR  abs/1203.5944 (2012)

\bibitem{DBLP:conf/isaac/CheongHKK13}
Cheong, O., Har-Peled, S., Kim, H., Kim, H.S.: On the number of edges of
  fan-crossing free graphs. In: Cai, L., Cheng, S.W., Lam, T.W. (eds.) ISAAC
  2013. LNCS, vol. 8283, pp. 163--173. Springer, Heidelberg (2013)

\bibitem{DBLP:journals/ijcga/DehkordiE12}
Dehkordi, H.R., Eades, P.: Every outer-1-plane graph has a right angle crossing
  drawing. Int. J. Comput. Geometry Appl.  22(6),  543--558 (2012)

\bibitem{DBLP:journals/tcs/DidimoEL11}
Didimo, W., Eades, P., Liotta, G.: Drawing graphs with right angle crossings.
  Theor. Comput. Sci.  412(39),  5156--5166 (2011)

\bibitem{DBLP:journals/tcs/EadesHKLSS13}
Eades, P., Hong, S.H., Katoh, N., Liotta, G., Schweitzer, P., Suzuki, Y.: A
  linear time algorithm for testing maximal 1-planarity of graphs with a
  rotation system. Theor. Comput. Sci.  513,  65--76 (2013)

\bibitem{DBLP:journals/dam/EadesL13}
Eades, P., Liotta, G.: Right angle crossing graphs and 1-planarity. Discrete
  Applied Mathematics  161(7-8),  961--969 (2013)

\bibitem{eggleton}
Eggleton, R.: Rectilinear drawings of graphs. Utilitas Mathematica  29,  149 --
  172 (1986)

\bibitem{DBLP:journals/dm/FabriciM07}
Fabrici, I., Madaras, T.: The structure of 1-planar graphs. Discrete
  Mathematics  307(7-8),  854--865 (2007)

\bibitem{Fary}
F\'{a}ry, I.: On straight line representations of planar graphs. Acta Sci.
  Math. Szeged  11,  229 -- 233 (1948)

\bibitem{DBLP:journals/siamdm/FoxPS13}
Fox, J., Pach, J., Suk, A.: The number of edges in k-quasi-planar graphs. SIAM
  J. Discrete Math.  27(1),  550--561 (2013)

\bibitem{gj-cigtnpc-79}
Garey, M.R., Johnson, D.S.: Computers and Intractability: A Guide to the Theory
  of NP-Completeness. W. H. Freeman \& Co., New York, NY, USA (1979)

\bibitem{DBLP:journals/algorithmica/GrigorievB07}
Grigoriev, A., Bodlaender, H.L.: Algorithms for graphs embeddable with few
  crossings per edge. Algorithmica  49(1),  1--11 (2007)

\bibitem{gutwenger/mutzel:gd2000}
Gutwenger, C., Mutzel, P.: A linear time implementation of {SPQR}-trees. In:
  Marks, J. (ed.) GD 2000. LNCS, vol. 1984, pp. 77--90. Springer, Heidelberg
  (2001)

\bibitem{DBLP:conf/gd/HongEKLSS13}
Hong, S.H., Eades, P., Katoh, N., Liotta, G., Schweitzer, P., Suzuki, Y.: A
  linear-time algorithm for testing outer-1-planarity. In: Wismath, S., Wolff,
  A. (eds.) GD 2013. LNCS, vol. 8242, pp. 71--82. Springer, Heidelberg (2013)

\bibitem{DBLP:conf/cocoon/HongELP12}
Hong, S.H., Eades, P., Liotta, G., Poon, S.H.: F{\'a}ry's theorem for 1-planar
  graphs. In: Gudmundsson, J., Mestre, J., Viglas, T. (eds.) COCOON. LNCS, vol.
  7434, pp. 335--346. Springer, Heidelberg (2012)

\bibitem{DBLP:journals/corr/KaufmannU14}
Kaufmann, M., Ueckerdt, T.: The density of fan-planar graphs. CoRR
  abs/1403.6184 (2014)

\bibitem{DBLP:journals/jgt/KorzhikM13}
Korzhik, V.P., Mohar, B.: Minimal obstructions for 1-immersions and hardness of
  1-planarity testing. Journal of Graph Theory  72(1),  30--71 (2013)

\bibitem{NagamochiTR}
Nagamochi, H.: Straight-line drawability of embedded graphs. Technical Reports
  2013-005, Department of Applied Mathematics and Physics, Kyoto University
  (2013)

\bibitem{DBLP:conf/jcdcg/PachRT02}
Pach, J., Radoicic, R., T{\'o}th, G.: Relaxing planarity for topological
  graphs. In: Akiyama, J., Kano, M. (eds.) JCDCG 2002. LNCS, vol. 2866, pp.
  221--232. Springer, Heidelberg (2002)

\bibitem{DBLP:journals/combinatorica/PachT97}
Pach, J., T{\'o}th, G.: Graphs drawn with few crossings per edge. Combinatorica
   17(3),  427--439 (1997)

\bibitem{DBLP:journals/iwc/Purchase00}
Purchase, H.C.: Effective information visualisation: a study of graph drawing
  aesthetics and algorithms. Interacting with Computers  13(2),  147--162
  (2000)

\bibitem{MR0187232}
Ringel, G.: Ein {S}echsfarbenproblem auf der {K}ugel. Abh. Math. Sem. Univ.
  Hamburg  29,  107--117 (1965)

\bibitem{DBLP:journals/jgt/Thomassen88a}
Thomassen, C.: Rectilinear drawings of graphs. Journal of Graph Theory  12(3),
  335--341 (1988)

\end{thebibliography}
\end{document}